\documentclass{llncs}

\usepackage{amsfonts,amsmath,amssymb}
\usepackage{bibunits}
\usepackage{enumerate}

\newcommand{\ie}{\emph{i.e.}}
\newcommand{\noi}{\noindent}
\newcommand{\oc}{\gamma}
\newcommand{\mc}[1]{\mathcal{#1}}
\newcommand{\cH}{\mc H}
\newcommand{\cA}{\mc A}
\newcommand{\cR}{\rho}

\begin{document}

\title{Hypergraph covering problems motivated by genome assembly
  questions}

\author{Cedric Chauve\inst{1,2}, Murray Patterson\inst{3}, Ashok Rajaraman\inst{2,4}}

\institute{LaBRI, Universit\'e Bordeaux 1, Bordeaux, France
\and Department of Mathematics, Simon Fraser University, Burnaby (BC), Canada
\email{[cedric.chauve,arajaram]@sfu.ca} 
\and Centrum Wiskunde \& Informatica, Amsterdam, The Netherlands
\email{murray.patterson@cwi.nl}
\and International Graduate Training Center in Mathematical Biology, Pacific Institute for Mathematical Sciences, Vancouver (BC), Canada}

\maketitle

\begin{abstract}
  The Consecutive-Ones Property (C1P) is a classical concept in
  discrete mathematics that has been used in several genomics
  applications, from physical mapping of contemporary genomes to the
  assembly of ancient genomes. A common issue in genome assembly
  concerns repeats, genomic sequences that appear in several locations
  of a genome. Handling repeats leads to a variant of the C1P, the C1P
  with multiplicity (mC1P), that can also be seen as the problem of
  covering edges of hypergraphs by linear and circular walks. In the 
  present work, we describe variants of the mC1P that address
  specific issues of genome assembly, and polynomial time or
  fixed-parameter algorithms to solve them.
\end{abstract}

\begin{bibunit}[plain]
\section{Introduction}\label{sec:introduction}


A binary matrix $M$ satisfies the Consecutive-Ones Property (C1P) if
its columns can be ordered in such a way that, in each row, all 1
entries appear consecutively.
The C1P has been studied in relation to a wide range of problems,
from theoretical computer science \cite{Dom2009} to genome mapping
(see \cite{OuangraouaTC2011,ZhangJZ2012} and references there).
The C1P can be naturally described in terms of covering hypergraph
edges by walks. Assume a binary matrix $M$ is the incidence matrix of
a hypergraph $H$, where columns represent vertices and rows encode
edges; then $M$ is C1P if and only if $H$ can be covered by a path
that contains all vertices and where every edge appears as a
contiguous subpath. Deciding if a binary matrix is C1P can be done in
linear time and space (see \cite{Dom2009} and references there). If a
matrix is not C1P, a natural approach is to remove the smallest number
of rows from this matrix in such a way that the resulting matrix is
C1P. This problem, equivalent to an edge-deletion problem on
hypergraphs that solves the Hamiltonian Path problem, is NP-complete,
although fixed-parameter tractability (FPT) results have recently been
published.


At a high level of abstraction, genome assembly problems can be seen
as graph or hypergraph covering problems: vertices represent small
genomic sequences, edges encode co-localisation information, and one
wishes to cover the hypergraph with a set of linear walks (or circular
walks for genomes with circular chromosomes) that respect
co-localisation information\footnote{Note to reviewers: we provide a
  more detailed description of the link between the assembly
  hypergraph framework and practical assembly problems in the
  appendix.}. Such walks encode the order of elements along
chromosomal segments of the assembled genome. One of the major issues
in genome assembly problems concerns {\em repeats}- genomic
elements that appear, up to limited changes, in several locations in
the genome being assembled. Such repeats are known to confuse
assembly algorithms and to introduce ambiguity in assemblies
\cite{TreangenS2012}.


Modeling repeats in graph theoretical models of genome assembly can be
done by associating to each vertex a {\em multiplicity}: the
multiplicity of a vertex is an upper bound on the number of
occurrences of this vertex in linear/circular walks that cover the
hypergraph, and thus a vertex with a multiplicity greater than $1$ can
traversed several times in these walks (\ie, encodes a repeat as 
defined above). This hypergraph covering problem naturally translates 
into a variant of the C1P, called the C1P with multiplicity (mC1P) 
that received little attention until recently, when it was investigated 
in several recent papers in relation to assembling ancestral genomes 
that describee both hardness and tractability results for decision and 
edge-deletion problems
\cite{BatzoglouI1999,WittlerMPS2011,ChauveMPW2011,ManuchPWCT2012}.


In the present paper, we formalize the previously studied C1P and mC1P
notions in terms of {\em covering of assembly hypergraphs} by linear
and circular walks and edge-deletion problems
(Section~\ref{sec:preliminaries}). Next, we describe new tractability
results for decision and edge-deletion problems
(Section~\ref{sec:results}): we show that deciding if a given assembly
hypergraph admits a covering by linear and circular walks that
respects the multiplicity of all vertices is FPT and we describe
polynomial time algorithms for decision and edge-deletion problems for
families of assembly hypergraphs which encode information allowing us
to clear ambiguities due to repeats.  We conclude with several open
questions (Section~\ref{sec:conclusion}).

\section{Preliminaries}\label{sec:preliminaries}


\subsection{Notation and terminology}\label{sec:terminology}

\begin{definition}\label{def:hypergraph}
  An {\em assembly hypergraph} is a quadruple $(H,w,c,o)$ where
  $H=(V,E)$ is a hypergraph and $w,\ c,\ o$ are three mappings such
  that $w:\ E \rightarrow \mathbb{R}$, $c:\ V \rightarrow \mathbb{N}$,
  $o:\ E \rightarrow V^*$ where $o(\{v_1,\dots,v_k\})$ is either a
  sequence on the alphabet $\{v_1,\dots,v_k\}$ where each element
  appears at least once, or $\lambda$ (the empty sequence).
\end{definition}

From now, we consider that $|V|=n$, $|E|=m$, $s=\sum_{e\in E}|e|$,
$\Delta=\max_{e\in E}|e|$, $\delta=\max_{v\in V}|\{e\in E\ |\ v\in e\}|$,
$\oc=\max_{v\in V}c(v)$. A vertex $v$ such that $c(v)>1$ is called a
{\em repeat}; $V_R$ is the set of repeats and $\cR=|V_R|$. Edges
s.t. $|e|=2$ are called {\it adjacencies}; from now, without loss of
generality, we assume that $o(e)=\lambda$ if $e$ is an adjacency.
Edges s.t. $|e|>2$ (resp. $|e=3|$) are called {\it intervals}
(resp. {\em triples}). We denote the set of adjacencies (resp. 
weights of adjacencies) by $E_A$ (resp. $w_A$)  and the
set of intervals (resp. weights of intervals) by $E_I$ (resp. $w_I$) . 
An interval is {\it ordered} if $o(e) \neq \lambda$; an assembly graph 
with no ordered interval is {\em unordered}. From now, unless 
explicitly specified, our assembly hypergraphs will be unordered 
and unweighted. We call $c(v)$ the {\em multiplicity} of $v$.

\begin{definition}\label{def:adjacencygraph} 
  An assembly hypergraph with no interval is an {\it adjacency
    graph}. Given an assembly hypergraph ${\cal H}=(H=(V,E),w,c,o)$,
  we denote its {\em induced adjacency graph} by 
  ${\cal H}_A=(H_A=(V,E_A),w_A,c,o_A)$\footnote{Note that $o_A(e)=\lambda$
    for every $e\in E_A$, as adjacencies are unordered.}.
\end{definition}

\begin{definition}\label{def:compatibility}
  Let $(H=(V,E),w,c,o)$ be an assembly hypergraph and $P$ (resp. $C$)
  a linear (resp. circular) sequence on the alphabet $V$. An unordered
  interval $e$ is {\em compatible} with $P$ (resp. $C$) if there is a
  contiguous subsequence of $P$ (resp. $C$) whose content is equal to
  $e$. An ordered interval $e$ is compatible with $P$ (resp. $C$) if
  there exists a contiguous subsequence of $P$ (resp. $C$) equal to
  $o(e)$ or its mirror.
\end{definition}

\begin{definition}\label{def:assembly}
  An assembly hypergraph $(H,w,c,o)$ admits a {\em linear assembly}
  (resp. {\em mixed assembly}) if there exists a set ${\cal A}$ of
  linear sequences (resp. linear and/or circular sequences) on $V$
  such that every edge $e \in E$ is compatible with at least one
  sequence of ${\cal A}$, and every vertex $v$ appears at most $c(v)$
  times in ${\cal A}$. The weight of an assembly is $\sum_{e\in
    E}w(e)$.
\end{definition}

An assembly as defined above can naturally be seen as a set of walks
(some possibly closed in mixed assemblies) on $H$ such that every edge
of $E$ is traversed by a contiguous subwalk. In the following, we
consider two kinds of algorithmic problems that we investigate for
different families of assembly hypergraphs and genome models, a
decision problem and an edge-deletion problem.
\begin{itemize}
\item The {\em Assembly Decision Problem}: Given an assembly
  hypergraph ${\cal H}=(H,w,$ $c,o)$ and a genome model (linear or
  mixed), does there exist an assembly of ${\cal H}$ in this model ?
\item The {\em Assembly Maximum Edge Compatibility Problem}: Given an
  assembly hypergraph ${\cal H}=(H=(V,E),w,c,o)$ and a genome model,
  compute a maximum weight subset $E'$ of $E$ such that the assembly
  hypergraph ${\cal H}'=(H'=(V,E'),\{w(e)\ |\ e\in E'\},c,\{o(e)\ |\
  e\in E'\})$ admits an assembly in this model.
\end{itemize}


\begin{definition}\label{def:repeatcluster}
  Let $(H=(V,E),w,c,o)$ be an assembly hypergraph. A {\em maximal
    repeat cluster} is a connected component of the hypergraph whose
  vertex set is $V_R$ and edge set is $\{e \cap V_R\ |\ e \in E\}$.
\end{definition}

As outlined in the introduction, vertices in an assembly hypergraph
represent genomic elements, each with an associated copy number
$c(v)$, while edges and their order (for intervals) encode
hypothetical co-localisation information, each with an associated
weight.  Linear and/or
circular sequences of vertices defining an assembly represent the
order of these genomic elements along chromosomal segments, the
circular ones representing circular chromosomes. A maximal repeat
cluster encodes a group of elements that are believed to appear in
several locations of the genome to assemble, although different
occurrences might differ in terms of content and/or order (see
\cite{QS2008} for example). Such repeated structures cause ambiguity
in genome assemblies based solely on adjacencies; for example, if
$V=\{a,b,c,d,e\}$, with $c(a)=c(b)=c(d)=c(e)=1$ and $c(c)=2$, and
$E=\{\{a,c\},\{b,c\},\{d,c\},\{e,c\}\}$, then there are essentially
three possible linear assemblies ($\{a.c.b,d.c.e\}, \{a.c.d,b.c.e\},
\{a.c.e,b.c.d\}\}$), while adding the ordered interval $\{a.c.d\}$
leads to a single possible assembly.

\subsection{Existing results}\label{sec:existing}

When no repeats are allowed ($\oc=1$), the Assembly Decision Problem
in the linear genome model is equivalent to asking if a binary matrix
has the C1P, which can be solved in $O(n+m+s)$ time and space. The set
of all linear assemblies can be encoded into a compact data structure,
the {\em PQ-tree}. In the mixed genome model, the problem can also be
solved in linear time, as it reduces to testing the circular C1P for
every connected component of the overlap graph of the matrix. The {\em
  PC-tree}, a slightly modified PQ-tree, can be used to encode all
mixed genome assemblies. We summarize some of these results in the
following theorem and refer to \cite{Dom2009} for a survey on these
questions.

\begin{theorem}\label{thm:cis1}
  The Assembly Decision Problem can be solved in $O(n+m+s)$ time and
  space when $\gamma=1$, in the linear and mixed genome models.
\end{theorem}

In the linear genome model, the Assembly Maximum Edge Compatibility
Problem is hard for adjacency graphs -- it solves the problem of
computing a set of paths that cover a maximum number of edges of the
graph -- but FPT results have recently appeared
\cite{DomGN2010,ZhangJZ2012}. Tractability results are less general
when repeats are allowed, as shown below.

\begin{theorem}\label{thm:ADadjacencies1}{\em \cite{WittlerMPS2011}}
  (1) The Assembly Decision Problem can be solved in time and space
  $O(n+m+s)$ for adjacency graphs ($\Delta=2$) in the linear and mixed
  genome models.  (2) In both genome models, the Assembly Decision
  Problem is NP-hard if $\Delta\geq3$ and $\oc\geq2$.
\end{theorem}
  
The principle of the proof for (1) is that an adjacency graph admits a
valid assembly if and only if every vertex has at most $2c(v)$
neighbours and, in the linear model, if every connected component $C$
satisfies $\sum_{v\in C} \deg(v)-2c(v)>0$. This result, combined with
the use of PQ-trees on the assembly hypergraph without its repeats, can
be extended slightly in the linear genome model.

\begin{theorem}\label{thm:ADadjacencies2}{\em \cite{ChauveMPW2011}}
  The Assembly Decision Problem can be solved in polynomial time and
  space in the linear genome model for unordered assembly hypergraphs
  where, for every edge $e$ containing a repeat, either $e$ is an
  adjacency or $e$ is an interval that contains a single repeat $r$
  and there exists an edge $e'=e\setminus\{r\}$.
\end{theorem}
  

Finally, to the best of our knowledge, the following is the 
only tractability result for edge-deletion problems when repeats are
allowed, limited to adjacency graphs and the mixed genome model.

\begin{theorem}\label{thm:AMEDadjacencies1}{\em \cite{ManuchPWCT2012}}
  (1) The Assembly Maximum Edge Compatibility Problem can be solved in
  polynomial time and space in the mixed genome model for adjacency
  graphs ($\Delta=2$). (2) The Assembly Maximum Edge Compatibility
  Problem is NP-hard in the mixed genome model if $\Delta\geq3$, even if 
  \(\oc=1\).
\end{theorem}

\section{New results}\label{sec:results}

We first show that the Assembly Decision Problem is FPT with respect
to parameters $\Delta,\delta,\oc$ and $\cR$. Then we describe positive
results for the case where the induced adjacency graph ${\cal H}_A$ is
assumed to admit an assembly and specific families of intervals are
added to clear ambiguities caused by repeats. We discuss the practical
implications of our positive results at the end of the section.

\subsection{The Assembly Decision Problem is fixed-parameter tractable}

\begin{theorem}
  The Assembly Decision Problem can be solved in space
  $O(n+m+s+\cR\oc)$ and time
  $O\left(\left(\delta(\Delta+\cR\oc)\right)^{2\cR\oc}\left(
      n+m+s+\cR\oc \right)\right)$ in the linear and mixed genome
  models.
\end{theorem}

\begin{proof}
  The principle of the proof is, for the given assembly hypergraph 
  \({\cal H}=\left(H,c\right)\)\footnote{Note that we do not consider $w$ and
    $o$ here as the weight does not impact decision problems and we
    deal with unordered hypergraphs. So, we eliminate both mappings 
    from our notation.}, to build another assembly hypergraph
  $\cH_f=(H_f,c_f)$ such that $c_f(v)=1$ for all
    $v\in V(\cH_f)$, by making $c(r)$ copies of each $r\in V_R$ and
    considering each possible set $f$ of choices of 2 neighbors for
    each of these copies. $\cH_f$ can then be checked for the
    existence of an assembly with Theorem \ref{thm:cis1}. The sets $f$
    of choices are made in such a way that $\cH$ has an assembly if
    and only if, for at least one of these sets $f$ of choices, ${\cal
      H}_f$ has an assembly. Finally, if $\Delta,\delta, \oc$ and
    $\cR$ are fixed, we prove that there is a fixed number of such 
    sets $f$.

  Let $R'(r)=\{r_i:1\leq i\leq c(r)\}$ be the set of copies we shall 
  introduce for each $r\in V_R$ (and $R'=\bigcup_{r\in V_R} R'(r)$), 
  $N(v)$ be the \emph{neighborhood} of $v$ in $H$, that is the set of 
  vertices belonging to edges containing $v$, and

  \begin{equation*}
    N'(r)=\{u\in V\setminus V_R:u\in N(r)\}\cup \bigcup_{p\in (V_R\cap
      N(r))\cup\{r\}} R'(p)
  \end{equation*}

  \noi be the ``new neighborhood'' from which we choose neighbors for
  vertices in $R'(r)$.  We represent each set of possible choices of 2
  neighbors\footnote{We consider only the case of 2 neighbors here for
    expository reasons; the complete proof, including the case of one
    or no neighbor, is similar.}  of each $r_i\in R'(r)$ with a
  mapping $f_r:R'(r)\rightarrow S_r$, 
  where $S_r=\{\{u,v\}:u,v\in N'(r)\}$. Let $f=\bigcup_{r\in V_R} f_r$
  be the collection of these mappings (itself a mapping
  $f:R'\rightarrow S'$ where $S'=\bigcup_{r\in V_R} S'_r$).

  We can now state the full algorithm as follows.
  \begin{enumerate}
    \item For each \(r\in V_{R}\), make $c(r)$ copies of \(r\), which
    defines the set \(R'(r)\). Let \(R'=\bigcup_{r}R'(r)\).
    \item For each \(v\in R'(r)\), choose \(2\) neighbours from \(N'(r)\), 
      thus defining $f_{r}$ for every \(r\in V_{R}\). This also defines 
      $f$ as the collection of mappings $f_{r}$ over all $r\in V_{R}$.
    \item Construct a new assembly hypergraph $\cH_f=(H_f=(V_f,E_f),c_f)$
      with $V_f=(V\setminus V_R)\cup R'$, $c_f(v)=1$ for all $v\in V_{f}$,
      and $E_{f}$ defined as follows: (1) for each $r_i\in R'(r)$, $r\in
      V_R$, $f(r_i)=\{u,v\}$ for some $u,v\in N'(r)$, add $\{r_i,u\}$ and
      $\{r_i,v\}$ to $E_f$ ($f$-edges) and (2) for each $e\in E$, add an
      edge $e'\in E_f$ containing $\{v:v\in e\setminus V_R\}$.
    \item For each $v\in V_f\setminus R'$ adjacent to a vertex of $r_1 \in
      R'$, let $v.r_1.\ \dots\ .r_k.u$ be the unique path in $H_f$ s.t.
      $\{r_1,\dots,r_k\}\subseteq R'$ and $u\in V_f\setminus R'$. Add all
      of $\{r_1,\dots,r_k\}$ to $e'$ for each $e'\in E_f$ such that $v\in
      e'$.
    \item Use Theorem~\ref{thm:cis1} on $\cH_{f}$.  Output Yes and exit if
      $\cH_f$ admits an assembly in the chosen genome model.
    \item Iterate over all possible sets of neighbour choices $f$ in Step
      2.
    \item Output No if no $\cH_{f}$ admits an assembly in the chosen genome 
      model.
  \end{enumerate}


  \paragraph{Algorithm correctness.}
  The premise for the algorithm is the following claim, which we state 
  and prove below.
  \begin{claim} 
    $\cH$ has an assembly if and only if, for some $f$, $\cH_f$ has an
    assembly.
    \label{lem:mapping}
  \end{claim}

  
  First, if $\cH$ has the assembly $\cA$, in $\cA$, we replace each
  occurrence of a vertex $r\in V_R$ by copies $r_i\in R'(r)$ where
  $R'(r) = \{r_i:1\leq i\leq c(r)\}$.  Let this new assembly be called
  $\cA'$.  Each such $r_i$ is adjacent to at most 2 other distinct
  vertices. We consider the mapping $f$ which maps each such $r_i$ to
  its two neighbours in this assembly $\cA'$.  If we can establish
  that the hypergraph obtained from this mapping and the new edges we
  introduce admits $\cA'$ as an assembly, we are done.
  
  To decide if $\cH_f$ has an assembly, we first note that any set of
  covering walks on $\cH_f$ is a set of paths (we cannot visit the
  same vertex twice because $c_f(v)=1$ for all $v\in V_f$).  Since $\cA$
  is a covering walk of $\cH$, by splitting the vertices of $V_R$ into
  distinct copies, we ensure that no vertex of $\cH_f$ is visited
  twice by $\cA'$.  Now, let us look at the set of edges $E_f$.  If all
  of them are covered as contiguous subsequences in $\cA'$, we are
  done.  We show this by the following observations.
  \begin{enumerate}
  \item In $\cA$, every edge $e$ occurs as a contiguous subsequence.
    Let $e'$ be the edge in $\cH_f$ corresponding to $e$. Then, by
    definition of $\cA'$, $e'$ must occur in it as a contiguous
    subsequence.
  \item For each $r_i\in R'(v)$ for some $r\in V_R$, we
    defined $f(r_i)=\{u,v\}$ using the assembly $\cA$.  So, we
    definitely get both adjacencies $\{r_i,u\},\{r_i,v\}$ in $\cA'$.
  \end{enumerate}
  So, $\cA'$ must be an assembly for $\cH_f$, which implies that
  $\cH_f$ has an assembly.
  
  
  Conversely, if the graph $\cH_f$ has an assembly, it contains all
  vertices $V\setminus V_R$, and occurrences of each $r_i \in R'(r)$
  for all repeat vertices $r\in V_R$.  If we remove the subscripts,
  \ie, $r_i$ becomes $r$ for all $i$, we get an assembly $\cA$, which
  we claim is an assembly for $\cH$, as $\cA$ will have the following
  properties.
  \begin{enumerate}
  \item Every vertex $v\in V$ appears at least once, and at most
    $c(v)$ times.
  \item For every edge $e'\in E$ consisting only of vertices in
    $V\setminus V_R$, we get a contiguous occurrence of $e\in E$,
    which is the corresponding edge in $\cH$.
  \item For every edge $e\in E$, such that $r\in e$ for some $r\in
    V_R$, there is an edge $e'\in E_f$ such that $r_i\in R'(r)$ has two
    neighbours and $r_i\in e'$. In this case, we get a contiguous
    occurrence of $e'$ including $r_i$. Removing the subscripts gives
    us a contiguous occurrence of $e$ in the new assembly $\cA$.
  \end{enumerate}
  So, $\cA$ contains occurrences of every edge $e\in E$ in $\cH$ as
  contiguous subsequences, which proves that $\cA$ is an assembly for
  $\cH$.  This proves the claim.

  This proof holds for both genome models as Theorem \ref{thm:cis1}
  considers them both.

  \paragraph{Algorithm complexity.}  The space complexity
  follows obviously from the construction of $\cH_f$. The choice of
  neighbours can be made in at most
  $\binom{\delta\left(\Delta+\cR\oc-1\right)}{2}$ ways for each new
  vertex. So, in total, we get at most
  $\binom{\delta\left(\Delta+\cR\oc-1\right)}{2}^{\cR\oc}$ possible
  mappings $f:R'\rightarrow S'$. The procedure on each $v_{i}$ can be
  done in time $O\left(1\right)$, since we just need to check its
  neighbours, which are at most $2$. Doing so for all vertices in $V_f$
  takes time at most $O\left(n+\cR\oc\right)$. The final step, checking
  for the existence of an assembly for a given $\cH_f$, can be done in
  $O\left(\left(n+\cR\oc\right)+\left(m+2\cR\oc\right)+s\right)$ time,
  since we add at most $2\oc \cR$ new edges, and $\cR\oc$ new
  vertices.  \qed
\end{proof}

\subsection{An edge-deletion algorithm for unordered intervals of size 3}

Now, we assume we are given an assembly hypergraph ${\cal
  H}=\left(H,w,c,o\right)$ whose induced adjacency graph ${\cal H}_A$
is known to have a mixed assembly. To state our result, we extend
slightly the notion of compatibility: an unordered interval $e$ is
said to be {\em compatible} with ${\cal H}_A$ if there exists a walk
in $H_A=(V,E_A)$ whose vertex set is exactly $e$.  We consider 
the interval compatibility problem defined below.

\smallskip\noindent The {\em Assembly Maximum Interval Compatibility
  Problem}: Given an assembly hypergraph ${\cal H}=(H=(V,E),w,c,o)$
such that ${\cal H}_A$ admits a mixed assembly, compute a maximum
weight subset of $E_I$, $S \subseteq E_I$, such that ${\cal H}'=(H'=(V,E'=E_A
\cup S),\{w(e)\ |\ e\in E'\},c,\{o(e)\ |\ e\in E'\})$ admits a mixed
assembly.


\begin{theorem}\label{thm:maxcompatibility1} 
  Let \({\cal H}=\left(H=\left(V,E\right),w,c,o\right)\) be a weighted
  assembly hypergraph such that ${\cal H}_A$ admits a mixed
  genome assembly, and each interval is a triple containing
  at most one repeat and compatible with ${\cal H}_A$. The Assembly
  Maximum Interval Compatibility Problem in the mixed genome model can
  be solved for \({\cal H}\) in linear space and $O((n+m)^{3/2})$
  time.
\end{theorem}

\begin{proof}
  
  The proof proceeds in two stages: we first show that repeat-free
  triples, as well as triples whose non-repeat vertices form an
  adjacency, must always be included in a maximum weight compatible
  set of triples. Then, we present an algorithm which uses the
  adjacency compatibility algorithm of Ma{\v{n}}uch et
  al.~\cite{ManuchPWCT2012} to decide which of the remaining triples
  to include. From now, we denote by $S$ a maximum weight subset of
  $E_I$ such that $(H'=(V,E_A \cup S),\{w(e)\ |\ e\in E_A \cup
  S\},c,\{o(e)\ |\ e\in E_A \cup S\})$ admits a mixed assembly.


  \begin{claim}
    If a triple $e\in E_I$ satisfies
    $e=\left\{v_{0},v_{1},v_{2}\right\}$, with
    $c\left(v_{0}\right) = c\left(v_{1}\right) =
    c\left(v_{2}\right) = 1$, then $e \in S$.
  \end{claim}
	
  As \(e\) is assumed to be compatible with ${\cal H}_A$ by
  hypothesis, there is a walk on these three vertices in \(H_A\). As a
  walk on three non-repeat vertices is a path, w.l.o.g we assume that
  the adjacencies in the path are \(\left\{v_{0},v_{1}\right\}\) and
  \(\left\{v_{1},v_{2}\right\}\) (the argument holds by symmetry for
  the other cases). Then, in any mixed assembly of \({\cal H}_A\), in
  order to contain both adjacencies, and to make sure that \(v_{1}\)
  appears exactly once in the assembly, the assembly must contain
  \(e\), in the order \(v_{0}.v_{1}.v_{2}\). So, it must be included 
  in $S$, as $S$ is a maximum weight subset of $E_{I}$.


  \begin{claim}
    If a triple $e\in E_I$ satisfies $e=\left\{v_{0},v_{1},r\right\}$,
    with $c\left(v_{0}\right) = c\left(v_{1}\right) = 1$,
    $c(r)>1$ and $\{v_0,v_1\} \in E_A$, then $e \in S$.
  \end{claim}

  For the triple \(e\) to be compatible with ${\cal H}_A$, \(r\) needs
  to be adjacent to at least one of \(v_{0}\) and \(v_{1}\). Assume,
  w.l.o.g, that \(\left\{v_{1},r\right\} \in E_A\). If \({\cal H}_A\)
  admits a mixed assembly, both \(\left\{v_{0},v_{1}\right\}\) and
  \(\left\{v_{1},r\right\}\) must occur in a path or a
  cycle. Furthermore, since \(c\left(v_{1}\right)=1\), these two
  adjacencies must occur in the same path or cycle, in the order
  \(v_{0}.v_{1}.r\). This is an occurrence of \(e\) as a contiguous
  sequence, which implies that such a triple must occur in every
  assembly of \({\cal H}\), and must be included in $S$.
  
  \smallskip We are now left with the set $E_I'$ of triples
  $e=\left\{v_{0},v_{1},r\right\}$ such that $r$ is a repeat and
  \(\left\{v_{0},v_{1}\right\}\notin E_A\), which means that \(r\) is
  adjacent to both \(v_{0}\) and \(v_{1}\), and we need to find a
  maximum weight subset of triples of this form. To do this, we rely
  on the optimal edge-deletion algorithm designed by Ma{\v{n}}uch et
  al.~\cite{ManuchPWCT2012} for adjacency graphs as shown below.
  
  \begin{enumerate}
  \item Initialize an empty set \(D\) and $E'=E_A$.
  \item For every \(e\in E_I'\):
    \begin{enumerate}[(i)] 
    \item Add an adjacency \(a_{e}=\left\{v_{0},v_{1}\right\}\) to
      \(D\), label $a_e$ with the triple $e$, and set
      \(w_D\left(a_{e}\right) = w\left(e\right)\).
    \item Remove $\{v_0,r\}$ and $\{v_1,r\}$ from $E'$, if present.
    \end{enumerate}
  \item For every remaining adjacency \(e\in E'\), set
    \(w'\left(e\right)=1+\sum_{a_{e}\in D}w_D\left(a_{e}\right)\).
  \item Apply the linearization algorithm (Theorem~\ref{thm:AMEDadjacencies1})
    ~\cite{ManuchPWCT2012} on \((H_D=(V,E' \cup D),w' \cup w_D,c,o_A)\).
  \item Add the triples corresponding to the labels of the adjacencies
    from $D$ retained by the linearization algorithm to \(S\).
  \end{enumerate}
  
  \noi{\em Algorithm correctness.} Given a triple
  \(e=\left\{v_{0},v_{1},r\right\}\) with a repeat vertex \(r\) and no
  adjacency \(\left\{v_{0},v_{1}\right\} \in E_A\), we consider a
  candidate mixed assembly of \({\cal H}\) containing the elements of
  \(e\) contiguously. In such an assembly, we would encounter the
  consecutive substring \(v_{0}.r.v_{1}\).  We can contract this
  substring and label the newly formed adjacency
  \(\left\{v_{0},v_{1}\right\}\), signifying that there is a path of
  length \(2\) between \(v_{0}\) and \(v_{1}\) which passes through
  \(r\) and contains no other vertices, \ie, it encodes the triple
  \(e\). So, we construct the new assembly hypergraph (an adjacency
  graph) by deleting the adjacencies \(\left\{v_{0},r\right\}\) and
  \(\left\{v_{1},r\right\}\) and encoding the path containing $e$ into
  the adjacency \(\left\{v_{0},v_{1}\right\}\) added to $D$.
  
  The optimal edge-deletion algorithm from \cite{ManuchPWCT2012}
  computes a maximum weight set of adjacencies \(S'\subseteq D\) such
  that the assembly graph \(\left(H_{opt}=\left(V, E'\cup
  S'\right),\right.$ $\left.w_{opt}',c,o_{opt}'\right)\) has a mixed 
  assembly, where
  \(w_{opt}'\) and \(o_{opt}'\) are the restrictions of \(w'\) and
  \(o'\) to \(E'\cup S'\).  In this assembly, we can replace every
  \(a_{e}\in S'\) by the corresponding triple \(e\) and the two
  corresponding adjacencies from $E_A$. Note that none of the
  adjacencies from $E_A$ are discarded during linearization since they
  are weighted so that discarding any one would be suboptimal when
  compared to discarding the entire set of adjacencies from $D$. So
  the assembly obtained by this process will contain all the edges
  from \(E_A\), as well as a maximum weight set \(S\subseteq E_I\)
  such that every \(e\in S\) is present. This implies that we computed
  a maximum weight compatible set of triples from $E_I'$.
    
  \noi{\em Algorithm complexity.} Checking the compatibility of a
  triple \(e\) with ${\cal H}_A$ can be done in constant time, since
  we just need a \(3\)-step graph search from any vertex \(v\in e\),
  and proceed until we find a path connecting all \(3\) vertices in
  \(e\).  We can also check the number of repeats in \(e\) in constant
  time. To deal with triples from the set $E_I'$, the new assembly
  hypergraph can obviously be constructed in $O(n+m)$ time and space,
  and contains $n$ vertices and $O(m)$ edges. So the optimal
  edge-deletion algorithm is the main component of the process, and is
  based on a maximum weight matching algorithm of time complexity
  $O((n+m)^{3/2})$ \cite{ManuchPWCT2012}.  \qed
\end{proof}

Related to this theorem, we have the following corollary.
\begin{corollary}\label{thm:maxcompatibility2}
  Let \({\cal H}=\left(H=\left(V,E\right),w,c,o\right)\) be an
  assembly hypergraph such that ${\cal H}_A$ admits a mixed genome
  assembly, maximal repeat clusters are all of size $1$, and each
  interval is an unordered compatible triple. The Assembly Maximum 
  Interval Compatibility Problem in the mixed genome model can be solved 
  for \({\cal H}\) in linear space and $O((n+m)^{3/2})$ time.
\end{corollary}

\begin{proof}
  We already know that we can find a maximal weight compatible subset
  \(S\subseteq E_I\) if there is no \(e\in E_I\) containing more than
  \(1\) repeat. 

  We now show that for the current problem, a triple
  \(e=\left\{v_{0},r_{0},r_{1}\right\}\), where \(r_{0}\) and
  \(r_{1}\) are repeats, and \(c\left(v_{0}\right)=1\), can also be
  included in the set \(S\) if it is compatible with ${\cal H}_A$.

  Note that \(r_{0}\) and \(r_{1}\) cannot have an adjacency between
  them, since the size of a maximal cluster cannot exceed \(1\). So,
  for \(e\) to be compatible, the corresponding adjacencies will be
  \(\left\{r_{0},v_{0}\right\}\) and
  \(\left\{r_{1},v_{0}\right\}\). For \({\cal H}_A\) to have a
  mixed assembly which contains both adjacencies, the assembly must
  contain \(e\) in the order \(r_{0}.v_{0}.r_{1}\). This is a
  contiguous appearance of the elements of \(e\), and it must occur
  in every mixed assembly. It can thus be included in $S$. 
  Theorem~\ref{thm:maxcompatibility1} concludes the proof.  \qed
\end{proof}

\subsection{A decision algorithm for ordered repeat spanning intervals}


\begin{definition}\label{def:repeatinterval}
  Let $(H=(V,E),w,c,o)$ be an assembly hypergraph. An interval $e\in
  E_{I}$ is an {\em ordered repeat spanning interval} for a maximal
  repeat cluster $R$ if $e=\{u,v,r_1,\dots,r_k\}$ with $c(u)=c(v)=1$,
  $\{r_1,\dots,r_k\} \subseteq R$ and $o(e)=u.s.v$, where $s$ is a
  sequence on the set $\left\{r_{1},\ldots,r_{k}\right\}$, containing
  every element at least once. The subset of ordered repeat spanning
  intervals in $E_{I}$ is denoted by $E_{rs}$
\end{definition}

%

\begin{theorem}~\label{thm:repeat_spanning_compatibility} Let \({\cal
    H}=\left(H=\left(V,E\right),w,c,o\right)\) be an assembly
  hypergraph such that every repeat $r\in V_{R}$ is either contained
  in an adjacency, or it is contained in an interval $e\in E_{I}$ of
  one of the following forms.
  \begin{enumerate}
  \item $e$ is an ordered repeat spanning interval.
  \item $r$ is the only repeat in $e$, $e'=e\setminus\left\{r\right\}\in E$, 
    and $o\left(e\right)=o\left(e'\right)=\lambda$.
  \end{enumerate}
  The Assembly Decision Problem in the linear genome model can
  be solved for \({\cal H}\) in polynomial time and space.
\end{theorem}

\begin{proof}
  The basic idea of the proof is to realize the sequence $o(e)$ for 
  every repeat spanning interval $e\in E_{rs}$ by creating unique 
  copies of the repeats in $e$ and decreasing the multiplicity 
  accordingly. This leads to an assembly graph that can then be 
  checked using Theorem \ref{thm:ADadjacencies2}. Formally 
  we define an extended assembly hypergraph, ${\cal
    H}'=(H'=(V',E'),c',o')$, as follows (we omit $w$ from the 
    notation, since we are addressing a decision problem).
  \begin{enumerate}
  \item $V'=V$, $E'=E\backslash E_{rs}$, $c'=c$, $o'=o_A$, $D=\emptyset$.
  \item For every repeat spanning interval $e \in E_{rs}$.
    \begin{enumerate}
    \item Let $o(e)=o=u.r_1.\dots.r_k.v$, possibly $r_{i}=r_{j}$ for
      $i\neq j$ (the $r_i$ are repeats).

    \item For $i$ from $1$ to $k$ 
      \begin{enumerate}
      \item add a unique vertex $t_i$ to $V'$, with multiplicity
        $c'(t_i)=1$,
      \item add an adjacency $\{t_{i-1},t_{i}\}$ to $E'$ for $1<i\leq k$,
      \item decrease $c'(r_i)$ by 1.
      \end{enumerate}
    \item Add edges $\{u,t_{1}\}$ and $\{v,t_{k}\}$ to $E'$.
    \item If the adjacencies $\{u,r_1\}$ and $\{r_k,v\}$ are present, 
        add them to $D$.
    \end{enumerate}
    \item Check if the assembly hypergraph, ${\cal
      H}'=(H'=(V',E'\setminus D),c',o')$ admits a linear genome assembly using
    Theorem \ref{thm:ADadjacencies2}.
  \end{enumerate}

  \begin{claim}
    ${\cal H}$ admits a valid genome assembly in the linear genome
    model if and only if $c'(r)\geq 0$ for every repeat $r\in V$ and
    ${\cal H}'$ admits one.
  \end{claim}

  Assume \({\cal H}'\) admits an assembly ${\cal A}'$. By
  construction, every repeat $r$ of $V_R$ maps to a subset of $V'$
  composed of $r$ and the vertices added when reading occurrences of
  $r$ in the ordered repeat spanning intervals of $E_I$. For a repeat
  $r \in V_R$, let $\phi(r) \subseteq V'$ be this subset of $V'$ and
  $\phi^{-1}$ the inverse map. By construction, the adjacencies added
  to $E'$ when reading the order $o(e)$ of an interval $e$, when the
  inverse map is applied $\phi^{-1}$ to their vertices, define a walk
  in ${\cal H}$ corresponding exactly to $o(e)$, which allows us to
  unambiguously translate the set of linear walks on $H'$ defining
  ${\cal A}'$ into a set of linear walks ${\cal A}$ on $H$. This
  implies that every edge of $E$ is compatible with ${\cal A}$ (as
  defined in Def. \ref{def:compatibility}), and we only need to
  consider potential problems caused by multiplicities. Assume that
  for every repeat $r\in V$ one has $c'(r)\geq 0$ and that for every
  $v'\in V'$, $v'$ appears at most $c'(v')$ times in an assembly of
  \({\cal H}'\), \ie, exactly \(1\) time, since $c'\left(v'\right)=1$
  for all \(v'\in V'\). For a vertex $v\in V$ such that $c(v)=1$, by
  construction $c'(v)=c(v)$, so an assembly of ${\cal H}'$ also
  satisfies the constraints of an assembly of ${\cal H}$ for $v$. For
  a repeat $r\in V$, the number of occurrences of elements of
  $\phi(r)$ in ${\cal A}'$ is at most
  $c'(r)+|\phi(r)\setminus\{r\}|$. By construction,
  $c(r)=c'(r)+|\phi(r)\setminus\{r\}|$, so assuming that $c'(r)\geq 0$
  implies that the constraint on $c(r)$ is satisfied in the
  linear walks on ${\cal A}$.

  Now, consider ${\cal H}$ admits an assembly ${\cal A}$ in the linear
  genome model. By definition, for every repeat spanning interval $e$,
  $o(e)$ appears as a walk in ${\cal A}$. By replacing the repeats in
  such a walk by new vertices with multiplicity $1$ as done in step
  2.b of the algorithm above, one clearly obtains an assembly ${\cal
    A}'$ for ${\cal H}'$, and the identity
  $c(r)=c'(r)+|\phi(r)\setminus\{r\}|$ ensures that $c'(r)\geq 0$.


  \noi{\em Complexity.} The polynomial time and space complexity
  follows from Theorem~\ref{thm:ADadjacencies2}, since the
  the construction of ${\cal H}'$ results in an assembly hypergraph
  with the structure in which no two repeats are contained in an
  interval (the repeat spanning intervals being resolved), and if an
  interval \(e\in E'\) contains a repeat \(r\), there exists an edge
  $e\setminus\left\{r\right\}$ in $E'$, since we added them directly
  from ${\cal H}$. \qed
\end{proof}

The following corollary follows easily from the previous theorem.
\begin{corollary}\label{cor:repeatspanning}
  Let \({\cal H}=\left(H=\left(V,E\right),w,c,o\right)\) be an
  assembly hypergraph such that each interval is an ordered repeat
  spanning interval. The Assembly Decision Problem in the mixed and
  linear genome models can be solved for \({\cal H}\) in
  $O(n+m+e+\sum_{e\in E_I}|o(e)|)$ time and space.
\end{corollary}
\begin{proof}
  We make the same construction as in
  Theorem~\ref{thm:repeat_spanning_compatibility}.  The extended
  assembly graph ${\cal H'}$ we create now is composed entirely of
  adjacencies, since $E_{I}=E_{rs}$. An application of
  Theorem~\ref{thm:ADadjacencies1} completes the proof. The time and
  space complexities follow immediately from the linear time and space
  complexities stated in Theorem~\ref{thm:ADadjacencies1} and from the
  size of ${\cal H'}$.  \qed
\end{proof}

The results above have interesting practical implications that we
outline now.  First, Corollary \ref{cor:repeatspanning} shows that, 
if provided with ordered repeat spanning intervals, one can check for 
the existence of an assembly in both genome models. Ordered repeat
spanning intervals can be obtained in practice in several ways, such
as mapping the elements of $V$ onto related genomes
\cite{RajaramanTC2013,GnerreLLJ2009} or long reads (see Appendix for
more details).  The tractability of the Assembly Decision Problem,
with linear time and space complexities, makes it possible to combine
it with the tractability result of Theorem \ref{thm:AMEDadjacencies1}
to select a subset of adjacencies, followed by a greedy heuristic for
the Assembly Maximum Interval Compatibility Problem. Note also that
the condition on the unordered intervals in the statement of Theorem
\ref{thm:repeat_spanning_compatibility} allows one to account for the
important notion of {\em telomeres} \cite{ChauveMPW2011}. Regarding
Theorem \ref{thm:maxcompatibility1}, it can be used to partially clear
the ambiguities caused by repeats in assembly hypergraphs where
triples are obtained from mate-pairs of reads from sequencing
libraries defined with inserts of length greater than the length of 
repeats\cite{NagarajanP2009}. If all maximal repeat clusters are
``collapsed'' into a single vertex (with the maximum multiplicity
among all initial repeats of the cluster), such mate-pairs spanning
repeat clusters define the triples. Solving the Assembly Maximum
Interval Compatibility Problem allows us to specify the locations of
the different occurrences of the spanned repeat clusters in the
assembled genome, thus resolving part of the ambiguity due to repeats
and leaving only the internal structure of each repeat cluster
(content and order) unresolved.

\section{Conclusion}\label{sec:conclusion}

In the present work, we presented a set of positive results on some
hypergraph covering problems motivated by genome assembly
questions. To the best of our knowledge, these are the first such
results for handling repeats in assembly problems in an edge-deletion
approach, as previous results focused on superstring approaches
\cite{KececiogluM1995,BatzoglouI1999,MedvedevGMB2007,NagarajanP2009},
and these new methods have been applied on real data
\cite{RajaramanTC2013}. Moreover, the initial results we presented
suggest several open problems.

First, our results about triples assume that they are compatible with
${\cal H}_A$ (\ie, appear as walks in $H_A$); we conjecture that
similar positive results can be obtained when relaxing this condition
(in particular when triple elements might not appear in the same
connected component). Next, our edge-deletion positive results assume
that ${\cal H}_A$ admits a genome assembly, and only intervals are
considered for being deleted. This leads to a two-stage assembly
process where adjacencies are deleted first, followed by intervals. It
remains open to see if both adjacencies and limited families of
intervals can be considered jointly. Also of interest would be to see
if the size of maximal repeat clusters or of intervals can be used as
parameters for FPT results.

Regarding repeat-spanning intervals, it can be asked if one can relax
the total order structure $o$ to account for uncertainty; for example,
if they are defined from the comparison of pairs of related genomes,
it might happen that specific rearrangements lead to conserved genome
segments that can be described by partial orders \cite{ZhengS2005},
which opens the question of solving the Assembly Decision Problem with
partial orders to describe repeat-spanning intervals. Along the same
line, it might happen that intervals spanning only prefixes or
suffixes of repeat occurrences (called {\em repeat-overlapping
  intervals}) can be detected, and the tractability of the Assembly
Decision Problem with such intervals is open; we conjecture it is FPT
in the number of such intervals.

Finally, {\em gaps}, that can be described in terms of binary
matrices, as entries $0$ appearing between entries $1$, appears
naturally in genome scaffolding problems \cite{GaoSN2009}; the notion
of gaps can naturally be described, for graphs, in terms of {\em
  bandwidth} and has been extended to binary
matrices/hypergraphs in \cite{ManuchPC2012}. Very limited tractability
result exist when gaps are allowed, whether it is for graphs
\cite{GaoSN2009} or hypergraphs \cite{ManuchPC2012}, none
considering repeats, which opens a wide range of questions of
practical importance.

\bibliographystyle{plain}

\begin{thebibliography}{10}

\bibitem{BatzoglouI1999}
S.~Batzoglou and S.~Istrail.
\newblock Physical mapping with repeated probes: The hypergraph superstring
  problem.
\newblock In {\em CPM}, volume 1645 of {\em LNCS}, pages 66--77, 1999.

\bibitem{ChauveMPW2011}
C.~Chauve, J.~Manuch, M.~Patterson, and R.~Wittler.
\newblock Tractability results for the consecutive-ones property with
  multiplicity.
\newblock In {\em CPM}, volume 6661 of {\em LNCS}, pages 90--103, 2011.

\bibitem{Dom2009}
M.~Dom.
\newblock Algorithimic aspects of the consecutive-ones property.
\newblock {\em Bull. EATCS}, 98:27--59, 2009.

\bibitem{DomGN2010}
M.~Dom, J.~Guo, and R.~Niedermeier.
\newblock Approximation and fixed-parameter algorithms for consecutive ones
  submatrix problems.
\newblock {\em J. Comput. Sys. Sci.}, 76:204--221, 2010.

\bibitem{GaoSN2009}
S.~Gao, W.-K. Sung, and N.~Nagarajan.
\newblock Opera: Reconstructing optimal genomic scaffolds with high-throughput
  paired-end sequences.
\newblock {\em J. Comput. Biol.}, 18:1681--1691, 2011.

\bibitem{GnerreLLJ2009}
S.~Gnerre, E.~S. Lander, K.~Lindblad-Toh, and D.~B. Jaffe.
\newblock Assisted assembly: how to improve a de novo genome assembly by using
  related species.
\newblock {\em Genome Biol}, 10:R88, 2009.

\bibitem{KececiogluM1995}
J.~D. Kececioglu and E.~W. Myers.
\newblock Combinatorial algorithms for dna sequence assembly.
\newblock {\em Algorithmica}, 13:7--51, 1995.

\bibitem{ManuchPC2012}
J.~Manuch, M.~Patterson, and C.~Chauve.
\newblock Hardness results on the gapped consecutive-ones property problem.
\newblock {\em Discrete Appl. Math.}, 160:2760--2768, 2012.

\bibitem{ManuchPWCT2012}
J.~Manuch, M.~Patterson, R.~Wittler, C.~Chauve, and E.~Tannier.
\newblock Linearization of ancestral multichromosomal genomes.
\newblock {\em BMC Bioinformatics}, 13(Suppl. 19):S11, 2012.

\bibitem{MedvedevGMB2007}
P.~Medvedev, K.~Georgiou, E.~W. Myers, and M.~Brudno.
\newblock Computability of models for sequence assembly.
\newblock In {\em WABI}, volume 4645 of {\em LNCS}, pages 289--301, 2007.

\bibitem{NagarajanP2009}
N.~Nagarajan and M.~Pop.
\newblock Parametric complexity of sequence assembly: theory and applications
  to {N}ext {G}eneration {S}equencing.
\newblock {\em J. Comput. Biol.}, 16:897--908, 2009.

\bibitem{OuangraouaTC2011}
A.~Ouangraoua, E.~Tannier, and C.~Chauve.
\newblock Reconstructing the architecture of the ancestral amniote genome.
\newblock {\em Bioinformatics}, 27:2664--2671, 2011.

\bibitem{QS2008}
J.~A.~A. Quitzau and J.~Stoye.
\newblock Detecting repeat families in incompletely sequenced genomes.
\newblock In {\em WABI}, volume 5251 of {\em LNCS}, pages 342--353, 2008.

\bibitem{RajaramanTC2013}
A.~Rajaraman, E.~Tannier, and C.~Chauve.
\newblock {F}{P}{S}{A}{C}: Fast phylogenetic scaffolding of ancient contigs.
\newblock Submitted, 2013.

\bibitem{TreangenS2012}
T.~J. Treangen and S.~L. Salzberg.
\newblock Repetitive {D}{N}{A} and next-generation sequencing: computational
  challenges and solutions.
\newblock {\em Nature. Rev. Genet.}, 13:36--46, 2012.

\bibitem{WittlerMPS2011}
R.~Wittler, J.~Manuch, M.~Patterson, and J.~Stoye.
\newblock Consistency of sequence-based gene clusters.
\newblock {\em J. Comput. Biol.}, 18:1023--1039, 2011.

\bibitem{ZhangJZ2012}
C.~Zhang, H.~Jiang, and B.~Zhu.
\newblock Radiation hybrid map construction problem parameterized.
\newblock In {\em COCOA}, volume 7402 of {\em LNCS}, pages 127--137, 2012.

\bibitem{ZhengS2005}
C.~Zheng and D.~Sankoff.
\newblock Genome rearrangements with partially ordered chromosomes.
\newblock In {\em COCOON}, volume 3595 of {\em LNCS}, pages 52--62, 2005.

\end{thebibliography}


\begin{thebibliography}{10}

\bibitem{BatzoglouI1999}
S.~Batzoglou and S.~Istrail.
\newblock Physical mapping with repeated probes: The hypergraph superstring
  problem.
\newblock In {\em CPM}, volume 1645 of {\em LNCS}, pages 66--77, 1999.

\bibitem{ChauveGOT2010}
C.~Chauve, H.~Gavranovic, A.~Ouangraoua, and E.~Tannier.
\newblock Yeast ancestral genome reconstructions: the possibilities of
  computational methods {I}{I}.
\newblock {\em J. Comput. Biol.}, 17:1097--1112, 2010.

\bibitem{ChauveT2008}
C~Chauve and E.~Tannier.
\newblock A methodological framework for the reconstruction of contiguous
  regions of ancestral genomes and its applications to mammalian genomes.
\newblock {\em PLoS Comput. Biol.}, 4:e1000234, 2008.

\bibitem{Csuros2010}
M.~Csur\"os.
\newblock Count: evolutionary analysis of phylogenetic profiles with parsimony
  and likelihood.
\newblock {\em Bioinformatics}, 26:1910--1912, 2010.

\bibitem{DuanZDW2013}
J.~Duan, J.-G. Zhang, H.-W. Deng, and Y.-P. Wang.
\newblock Comparative studies of copy number variation detection methods for
  next-generation sequencing technologies.
\newblock {\em PLoS ONE}, 8:e59128, 2013.

\bibitem{Ma2006}
J.~Ma et~al.
\newblock Reconstructing contiguous regions of an ancestral genome.
\newblock {\em Genome Res.}, 16:1557--1565, 2006.

\bibitem{Koren2012}
S.~Koren et~al.
\newblock Hybrid error correction and de novo assembly of single-molecule
  sequencing reads.
\newblock {\em Nat. Biotechnol.}, 30:693--700, 2012.

\bibitem{GaoBN2012}
S.~Gao, D;~Bertrand, and N.~Nagarajan.
\newblock Fin{I}{S}: Improved in silico finishing using an exact quadratic
  programming formulation.
\newblock In {\em WABI}, volume 7534 of {\em LNCS}, pages 314--325, 2012.

\bibitem{GnerreLLJ2009}
S.~Gnerre, E.~S. Lander, K.~Lindblad-Toh, and D.~B. Jaffe.
\newblock Assisted assembly: how to improve a de novo genome assembly by using
  related species.
\newblock {\em Genome Biol}, 10:R88, 2009.

\bibitem{HusonRM2002}
D.~H. Huson, K.~Reinert, and E.~W. Myers.
\newblock The greedy path-merging algorithm for contig scaffolding.
\newblock {\em J. ACM}, 49:603--615, 2002.

\bibitem{IduryW1995}
R.~M. Idury and W.~S. Waterman.
\newblock A new algorithm for {D}{N}{A} sequence assembly.
\newblock {\em J. Comput. Biol.}, 2:291--306, 1995.

\bibitem{KececiogluM1995}
J.~D. Kececioglu and E.~W. Myers.
\newblock Combinatorial algorithms for dna sequence assembly.
\newblock {\em Algorithmica}, 13:7--51, 1995.

\bibitem{MedvedevGMB2007}
P.~Medvedev, K.~Georgiou, E.~W. Myers, and M.~Brudno.
\newblock Computability of models for sequence assembly.
\newblock In {\em WABI}, volume 4645 of {\em LNCS}, pages 289--301, 2007.

\bibitem{Myers1995}
E.~W. Myers.
\newblock Toward simplifying and accurately formulating fragment assembly.
\newblock {\em J. Comput. Biol.}, 2:275--290, 1995.

\bibitem{NagarajanP2009}
N.~Nagarajan and M.~Pop.
\newblock Parametric complexity of sequence assembly: theory and applications
  to {N}ext {G}eneration {S}equencing.
\newblock {\em J. Comput. Biol.}, 16:897--908, 2009.

\bibitem{OuangraouaTC2011}
A.~Ouangraoua, E.~Tannier, and C.~Chauve.
\newblock Reconstructing the architecture of the ancestral amniote genome.
\newblock {\em Bioinformatics}, 27:2664--2671, 2011.

\bibitem{QS2008}
J.~A.~A. Quitzau and J.~Stoye.
\newblock Detecting repeat families in incompletely sequenced genomes.
\newblock In {\em WABI}, volume 5251 of {\em LNCS}, pages 342--353, 2008.

\bibitem{RajaramanTC2013}
A.~Rajaraman, E.~Tannier, and C.~Chauve.
\newblock {F}{P}{S}{A}{C}: Fast phylogenetic scaffolding of ancient contigs.
\newblock Submitted, 2013.

\bibitem{WittlerS2009}
J.~Stoye and R.~Wittler.
\newblock A unified approach for reconstructing ancient gene clusters.
\newblock {\em IEEE/ACM Trans. Comput. Biology Bioinform.}, 6:387--400, 2009.

\bibitem{WittlerMPS2011}
R.~Wittler, J.~Manuch, M.~Patterson, and J.~Stoye.
\newblock Consistency of sequence-based gene clusters.
\newblock {\em J. Comput. Biol.}, 18:1023--1039, 2011.

\bibitem{ZerbinoB2008}
D.~R. Zerbino and E.~Birney.
\newblock Velvet: algorithms for de novo short read assembly using de bruijn
  graphs.
\newblock {\em Genome Res.}, 18:821--829, 2008.

\bibitem{ZhengS2005}
C.~Zheng and D.~Sankoff.
\newblock Genome rearrangements with partially ordered chromosomes.
\newblock In {\em COCOON}, volume 3595 of {\em LNCS}, pages 52--62, 2005.

\end{thebibliography}
\putbib[iwoca2013]
\end{bibunit}

\vfill\pagebreak
\begin{bibunit}[plain]
\section*{Appendix A.}

In this appendix, we describe how the assembly hypergraph relates to
practical genome assembly problems.

Our initial motivation for investigating the algorithmic problems
described in this paper follows from earlier computational
paleogenomics methods developed to compute genome maps and scaffolds
for ancestral genomes
\cite{Ma2006,ChauveT2008,WittlerS2009,ChauveGOT2010,OuangraouaTC2011,RajaramanTC2013}. In
this problem, the vertex set $V$ represents a set of $n$ ancestral
genomic markers, obtained either through whole genome alignment
\cite{Ma2006,ChauveT2008,OuangraouaTC2011}, the analysis of gene
families \cite{ChauveGOT2010}, or the sequencing of an ancient
genome \cite{RajaramanTC2013}. The function $c$ encodes the
multiplicity, that is an upper bound on the allowed number of copies
of each marker in potential assemblies. For ancestral genomes, it can
be obtained from traditional parsimony methods \cite{Csuros2010}.  An
edge $e=\{v_1,\dots,v_k\}\in E$ encodes the hypothesis that
$v_1,\dots,v_k$ appear {\it contiguously} in an assembly of the
elements of $V$. For ordered intervals, that are edges $e$, such that
$|e|>2$ and $o(e)\neq \lambda$, $o(e)$ encodes a total ordering
information about the genomic elements they contain. In computational
paleogenomics, edges and intervals (including order) can be obtained
from the comparison of pairs of genomes related to the ancient genome
that is being assembled. The function $w$ is a weight that can be seen
as a confidence measure on every edge (the higher, the better), that
can be based on phylogenetic conservation. More generally, the
assembly hypergraph is a natural model for genome mapping problems
\cite{BatzoglouI1999,ZhengS2005}.

However, the assembly hypergraph also allows us to formalize other
assembly problems. For example, in the {\em scaffolding} problem
\cite{HusonRM2002}, $V$ would represent {\em contigs} and $c$ can be
obtained by methods based on the reads depth of coverage
\cite{GaoBN2012,DuanZDW2013}. Co-localization information can be
obtained from mate-pairs libraries with an insert that is short with
respect to the minimum contig length, thus describing adjacencies,
while ordered intervals can be obtained from mapping contigs onto long
reads \cite{Koren2012} or related genome sequences
\cite{RajaramanTC2013,GnerreLLJ2009}.

The assembly hypergraph can also be used to model the problem of
assembling short reads into contigs, although contig assembly is
generally based on Eulerian superstring approaches
\cite{MedvedevGMB2007,NagarajanP2009,KececiogluM1995} instead of edge
deletions approaches. In this problem, the vertices $V$ represent
short sequence elements, such as reads in the overlap graph approach
\cite{Myers1995} or $k$-mers (substrings of length $k$) in the widely
used de Bruijn graph approach
\cite{IduryW1995,ZerbinoB2008}\footnote{For example, the notion of
  maximal repeat cluster is very similar to the notion of connected
  components of the sparse de Bruijn graph that was studied in
  \cite{QS2008}.}.  The function $c$ can here again be obtained from
the reads depth of coverage. Adjacencies follow from overlaps between
elements of $V$, whose statistical significance, combined with the
read quality for example, can be used to define $w$. Intervals can
here again be obtained from mapping short reads on long reads.

Finally, it is important to remember that genomic segments are {\em
  oriented} along a chromosome, due to the double stranded nature of
most genomes. The algorithms we described in the present paper can
handle this problem in a very easy way. Each genomic element is
represented by two vertices, one for each extremity, with an adjacency
linking them (called a {\em required} adjacency, while adjacencies
between extremities of different elements are called {\em inferred}
adjacencies). A compatible assembly then needs to be composed of linear 
 or circular walks where required adjacencies alternate with inferred
adjacencies. This property can be handled naturally by the decision
algorithms (see \cite{WittlerMPS2011}), and also by the optimization
algorithms by weighting each required adjacency by a weight greater
than the cumulative weight of all inferred adjacencies. Also, triples
that overlap repeats need to be replaced by quadruples containing both
extremities of a same initial genomic element, which can be handled by
our algorithms (full details will be given in the complete version of
our work).

\putbib[iwoca2013]
\end{bibunit}

\end{document}